\documentclass[pra,aps,twocolumn,floatfix]{revtex4}
\usepackage{amsmath}
\usepackage{graphicx}
\usepackage{dcolumn}
\usepackage{amssymb}
\usepackage{color}

\usepackage{mathtools}

\usepackage{graphicx,comment}
\usepackage{epstopdf}
\usepackage{bm}
\usepackage{enumerate}
\usepackage{fancyhdr,color}
\usepackage{verbatim}
\usepackage{amsmath,amsthm}
\usepackage{amsfonts}
\usepackage{hyperref}
\usepackage[caption=false]{subfig}
\usepackage{tikz}

\def\cH{\mathcal{H}}

\def\cQ{\mathcal{Q}}

\newtheorem{definition}{Definition}

\newtheorem{lemma}{Lemma}
\newtheorem{corollary}{Corollary}

\def\one{{\mathchoice {\rm 1\mskip-4mu l} {\rm 1\mskip-4mu l} {\rm
1\mskip-4.5mu l} {\rm 1\mskip-5mu l}}}
\newcommand{\ket}[1]{\left| #1\right\rangle}        
\newcommand{\sket}[1]{| #1\rangle}  
  
\newcommand{\bra}[1]{\left\langle #1\right|}        
\newcommand{\braket}[2]{\langle #1 | #2 \rangle} 
\newcommand{\ketbra}[1]{| #1 \rangle \! \mspace{2mu}\langle #1| }

\begin{document}

\title{Improved implementation of reflection operators}

\author{
Anirban Narayan Chowdhury}
\affiliation{Center for Quantum Information and  Control, University of New Mexico, Albuquerque, NM 87131, US}
\affiliation{New Mexico Consortium,
Los Alamos,
NM 87545, US.}

\author{
Yi\u{g}it Suba\c{s}\i}
\affiliation{Theoretical Division, Los Alamos National Laboratory, Los Alamos, NM 87545, US.}

\author{
Rolando Diego Somma}
\affiliation{Theoretical Division, Los Alamos National Laboratory, Los Alamos, NM 87545, US.}

\date{\today}	

\begin{abstract}
Quantum algorithms for diverse problems, including search and optimization problems,
require the implementation of a reflection operator over a target state. Commonly, such reflections
are approximately implemented using phase estimation. 
Here we use 
a linear combination of unitaries and a version of amplitude
amplification to approximate reflection operators over eigenvectors of unitary operators using exponentially less ancillary qubits
in terms of a precision parameter.
The gate complexity of our method is also comparable to that of the phase estimation approach
in a certain limit of interest.
%
%
Like phase estimation, our method requires the implementation of controlled unitary operations.
We then extend our results to the Hamiltonian case where the target state is an eigenvector of a Hamiltonian
whose matrix elements can be queried. 
Our results are useful in that they reduce the resources required by
various quantum algorithms in the literature. Our improvements
also rely on an efficient quantum algorithm to prepare a quantum state with Gaussian-like 
amplitudes that may be of independent interest. We also provide a lower bound
on the query complexity of implementing approximate reflection operators on a quantum computer.
\end{abstract}

\maketitle

\section{Introduction}
\label{sec:intro}
Large quantum computers will be able to solve
problems that may never be solved
by classical computers. There are numerous
examples of problems for which a quantum speedup exists -see Ref.~\cite{jordanzoo} for a summary.
Rather than investigating new problems, this paper
is concerned with improving the resources required by known
quantum algorithms for, e.g., adiabatic state transformations, search, and related optimization 
problems, including those in Refs.~\cite{MNR07,BKS10}.

A key procedure that is used in Refs.~\cite{MNR07,BKS10}
is that of performing a reflection (a unitary transformation) over a target quantum state $\ket{\psi_0} \in \cH$,
where $\cH$ is a  Hilbert space of dimension $D<\infty$.
This reflection may be used within the context of amplitude amplification, a well-known
technique for quantum algorithms~\cite{Gro96,BHMT02,Amb12}.
The quantum state has the property that it is the unique eigenvector of a unitary $U$ of eigenvalue 1 (eigenphase $\lambda_0=0$), i.e.,
\begin{align}
U \ket{ \psi_0} =e^{i \lambda_0}\ket{ \psi_0} = \ket { \psi_0} \;,
\end{align}
although more complicated cases can be analyzed similarly (e.g., when the degeneracy of the eigenvalue 1 is greater than one
or when $\lambda_0 \ne 0$ is known).
 Here we assume that there is a procedure to implement a controlled operation $U$ and $U^\dagger$. We will
 measure the query complexity of a quantum algorithm, $C_U$, as the number of times
 that $U$ and $U^\dagger$, or their controlled versions, have to be invoked to solve the problem, with sufficiently high probability and precision.
 The gate complexity, $C_B$, will be the number of additional
 two-qubit gates that are independent of $U$. 
  If $R_{ \psi_0} $ is the desired reflection, it  has to satisfy
 \begin{align}
 \label{eq:ref1}
 R_{ \psi_0} \ket { \psi_0}  &= \ket { \psi_0} \;,  \\
  \label{eq:ref2}
 R_{ \psi_0}  \sket{\psi^\perp} &= -\sket{\psi^\perp} \;,
 \end{align}
 where $\sket{\psi^\perp}$ is any quantum state that is orthogonal to $\ket{ \psi_0} $, i.e.,
 $\bra{ \psi_0}  \psi^\perp \rangle=0$.
 We are interested in using the (controlled) operations $U$ and $U^\dagger$ to implement $\tilde R_{ \psi_0} $ such that it is an $\epsilon$-approximation
 of $R_{ \psi_0} $, for a given precision parameter $\epsilon <1 $.

 Constructing $\tilde R_{ \psi_0} $ from $U$ and $U^\dagger$ is not trivial in general
 and may require using additional information about $U$.
 We will assume that there exists a known $\Delta >0$
 so that any other nonzero eigenphase of $U$ satisfies
 $\Delta \le \lambda_j \le 2\pi -\Delta$, $j > 0$. This additional information
 is also used in Refs.~\cite{MNR07,BKS10}
 and may not be too strong.
  A standard procedure to build $\tilde R_{ \psi_0} $ is then
  via the well known phase estimation algorithm (PEA)~\cite{PEA95}. Roughly,
  the steps to implement the reflection
  using the PEA are as follows: i) encode the value of the eigenphase in
  an ancillary $n$-qubit register, ii) perform a reflection over the ancillary state $\ket0=\ket 0^{\otimes n}$ (representing $\lambda_0=0$),
  and iii) implement the inverse of the operation in step i). 

Our main goal is to significantly improve the number of ancillary qubits
required by the above PEA approach, without increasing the query and gate complexities. 
To this end, we will use
two techniques considered recently
 within the context of Hamiltonian simulation~\cite{BCC+15}.
 One technique is based on a decomposition of a unitary operation
 as a linear combination of unitaries (LCU), which
 was also considered in other works~\cite{SOGKL02,CW12,CKS15,BCC+15}.
 The other technique is based on a version of amplitude amplification
 and we refer to it as oblivious amplitude amplification (OAA), which was also considered in Ref.~\cite{BCC+15}.
 We present a quantum algorithm that implements
 $\tilde R_{ \psi_0}$ and requires
  $n=O (  \log    \log(1/\epsilon) + \log(1/ \Delta  )  )$ ancillary qubits.
  We demonstrate that this is an improvement with respect
 to the PEA approach where the number of ancillary qubits is
 $n=O(\log(1/\epsilon)\log(1/\Delta))$. The gate complexity
 of our quantum algorithm is comparable to that of the PEA 
 approach in a limit of interest where $\Delta = O(1/\log(1/\epsilon))$.
 This limit includes cases where $\Delta \ll 1$ and $ \epsilon \ll 1$.
 The gate complexity is $O(\log(1/\Delta) \log(\log(1/\Delta)/\epsilon))$.
 We emphasize that our analysis of the gate complexity of the PEA approach leads to an improvement 
 which is almost quadratic in $\log(1/\Delta)$ over that stated in \cite{MNR07}. 
 This is made possible through use of the approximate quantum Fourier transform.
 The query complexity of both approaches is $C_U=O(\log(1/\epsilon)/\Delta)$.
 Our results are 
 therefore useful
 to reduce the number of ancillary qubits 
 required by various quantum algorithms in the literature.
 
 A similar problem was considered in Ref.~\cite{Tul16},
 where the author showed that $\tilde R_{ \psi_0}$ can be implemented
 using $O(\log(1/\Delta))$ ancillary qubits and $O(\log^2(1/\epsilon)/\Delta)$ queries.
 We consider our results to be an improvement of Ref.~\cite{Tul16} since the $\epsilon$-dependence 
 of the number of ancillary qubits in our approach is still too small and the query complexity is 
 the same as the PEA approach, i.e., smaller than that of Ref.~\cite{Tul16}.


 Our paper is organized as follows. 
 In Sec.~\ref{sec:problem} we formalize the problem. 
 In Sec.~\ref{sec:PEA} we discuss the resources needed to perform reflections
  using the PEA approach and analyze the case where the approximate quantum Fourier transform is considered. In Sec.~\ref{sec:LCU}
 we describe our quantum algorithm (the LCU approach) based on the Poisson summation formula
 and the approximation as a LCU. The correctness of the LCU approach is discussed in Sec.~\ref{sec:LCUcorrect}.
 The operations involved in our quantum algorithm are
 discussed in detail in Secs.~\ref{sec:WB} and~\ref{sec:selectU}.
 The resources required by our approach are discussed in  Sec.~\ref{sec:LCUcomplex}.
 In Sec.~\ref{sec:Hamilt} we extend our results to the Hamiltonian case, which is basically given when $U=e^{i(H-\lambda_0)}$,
 for some Hamiltonian $H$. To obtain the complexities and number of ancillary qubits needed in this case we
 resort to a recent method for Hamiltonian simulation in Ref.~\cite{LC17}.
 Finally, we obtain a lower bound on the query complexity in Sec.~\ref{sec:lowerbound}.

\section{Problem statement}
\label{sec:problem}
Let $U$ and $R_{\psi_0}$ be the unitaries of Sec.~\ref{sec:intro}.
Our goal is to construct a quantum algorithm that implements a unitary operation $A$ that approximates a reflection over
a target state $\ket{\psi_0}$ as follows: 
\begin{align}
\label{eq:statement}
\| A \ket 0 \ket \xi - \ket 0 (R_{\psi_0} \ket \xi) \| \le \epsilon \;.
\end{align}
$\epsilon >0$ is a given error parameter, $\ket \xi \in \cH$ is any quantum state, and $\ket 0$ is an ancillary
state of $n$ qubits. $\|.\|$ is the Euclidean norm.
 We assume that there is a mechanism to implement controlled unitaries $U$ and $U^\dagger$ and
other $U$-independent quantum gates.
The query complexity of our quantum algorithm is measured by the number
of controlled-$U$ and their inverses needed to apply $A$. The gate complexity is the number of additional
two-qubit gates (independent of $U$) to implement $A$. These two-qubit gates are assumed to be
exactly implemented (i.e., we do not invoke any error correction or approximation method such as those discussed in Ref.~\cite{BGS13}).

The Hamiltonian version of this problem is defined such that
$\ket{\psi_0}$ is an eigenstate of a Hamiltonian $H$ where $\| H \| \le 1$. In this case
we need a mechanism $U'$ that implements an approximation
of $U:=e^{i(H-\lambda_0)}$ so
we can reduce this problem to the one described above. Here, $\lambda_0$ is the eigenvalue of $H$
corresponding to $\ket{\psi_0}$ and is assumed to be known.
We consider the scenario where $H$  is $d$-sparse and its matrix elements
can be queried. That is, any row (column) of the
$D \times D$ matrix $H$ has at most $d$ nonzero 
matrix elements. We assume that there is a procedure
$\cQ_H$ that computes such elements as follows.
For any ${\rm j} \in \{1,2,\ldots,D\}$ and ${\rm l} \in \{1,2,\ldots,d\}$,
\begin{align}
\cQ_H \ket{{\rm j,l}} \rightarrow \ket{{\rm j,v(j,l)}} \;,
\end{align}
where  ${\rm v(j,l)}$ is the row index of the l-th nonzero
matrix element in the j-th column of $H$.
The procedure $\cQ_H$ also allows us to perform the transformation
\begin{align}
\cQ_H \ket{{\rm j,k,z}} \rightarrow \ket{{\rm j,k,z\oplus h_{jk}}}
\end{align}
for any ${\rm j,k }\in \{1,2,\ldots,D\}$. ${\rm h_{jk}}$ is the corresponding
matrix element of $H$ (assumed to be described within 
a fixed number of bits $h$) and ${\rm z }\in \{0,1\}^h$.
This procedure is used in several recent methods for Hamiltonian
simulation (c.f., Ref.~\cite{LC17}). References~\cite{BCC+14,BCC+15,LC17}
describe a way to construct $U'$ using $\cQ_H$, satisfying 
\begin{align}
\| U'\ket 0 \ket \xi- \ket 0 U \ket \xi \| \le \varepsilon \;,
\end{align}
for any $0< \varepsilon < 1$ and any $\ket \xi$. $\ket 0$ is an ancillary state of $n'_H \ge 1$ qubits.

\section{The PEA approach}
\label{sec:PEA}
For completeness,
we provide a method to perform approximate reflections on the target state $\ket {\psi_0}$ using the PEA.
This method was used in Refs.~\cite{MNR07,PW09,BKS10}.
The PEA is depicted in Fig.~\ref{fig:pea} 
and is represented by a system-ancilla
unitary $W$.  $F^{\rm d}$ is the well-known quantum Fourier transform~\cite{PEA95}.
We consider the case where the input state $\sket {\psi_j}$ is an eigenstate of $U$ of eigenphase $\lambda_j$.
The quantum state prepared right before the action of $(F^{\rm d})^\dagger$ is ($M=2^n$)
\begin{align}
\frac 1 {\sqrt M }\sum_{m=0}^{M-1} \ket m e^{i m \lambda_j} \ket {\psi_j} \;.
\end{align}
If $\lambda_j= 0$ (i.e., $j=0$), the action of $(F^{\rm d})^\dagger$ transforms the state exactly 
back to $\ket 0 \sket {\psi_0}= \ket 0^{\otimes n} \sket{\psi_0}$.
We let $R:=2 P -\one$ be the simple reflection operator over the ancilla state $\ket 0$ and $P :=\ketbra 0 \otimes \one$ is the 
projector. We also define the unitary $A:= W^\dagger RW$.

\begin{figure}
\begin{center}
\includegraphics[width=.47\textwidth]{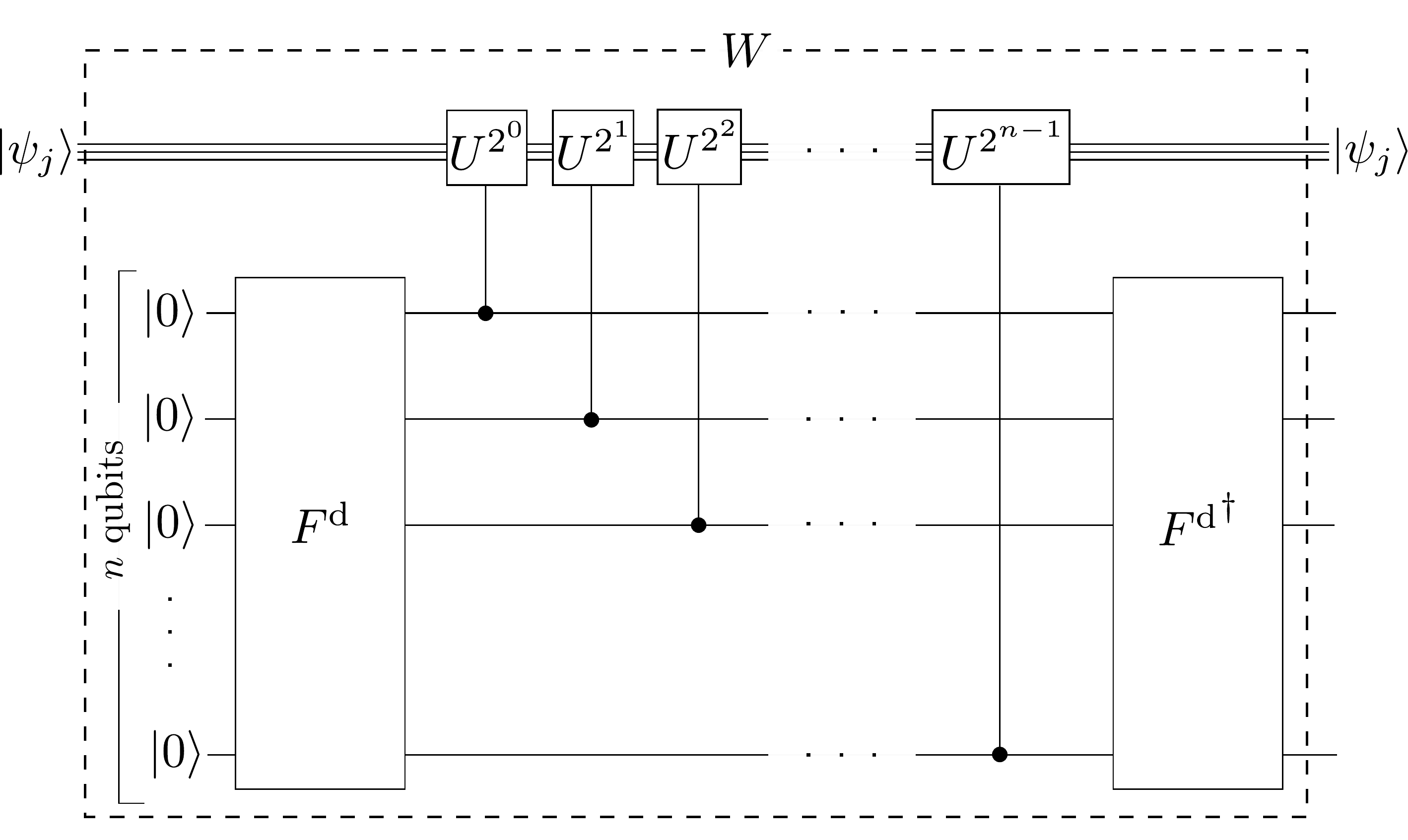}
\end{center}
\caption{Phase estimation algorithm. The unitary $W$ can be used to construct an approximate
reflection over the eigenstate $\ket{\psi_0}$ of the unitary $U$.}
\label{fig:pea}
\end{figure}

The state $\ket 0 \ket{\psi_0}$ is invariant if we act with $A$ because $R\ket 0 \ket {\psi_0} =\ket 0 \ket {\psi_0}$.
The situation is different if $j \ne 0$ and $\lambda_j >0$. It is well-known  (c.f.,~\cite{CEMM98,KOS07}) 
that if $\Delta \le \lambda_j \le 2\pi-\Delta$,
choosing $n=O(\log(1/\Delta))$ suffices to satisfy
\begin{align}
\label{eq:PEAaction}
W \ket 0 \sket{\psi_j} =  \sqrt p \ket 0 \sket{\psi_j} + \sqrt{1-|p|} \sket{0^\perp} \ket{\psi_j} \;.
\end{align}
Here, $ \sket{\psi_j}$ is the eigenvector with $j > 0$ and
  $\sket{0^\perp}$ is a quantum state that has support in the subspace of the ancilla 
orthogonal to $\ket 0$. $|p|$ is a constant that satisfies, e.g., $|p| < 1/16$.

Applying $A$ to  $\ket 0 \sket{\psi_j}$
approximates then the desired reflection with  constant  approximation error.
To see this, we note that
\begin{align}
\bra 0 \bra{\psi_j} W^\dagger \ket 0 \ket{\psi_j} &= \sqrt{p}^* \;,\\
\bra 0 \bra{\psi_j} W^\dagger \sket {0^\perp} \ket{\psi_j} &= \sqrt{1-|p|} \;,
\end{align}
and then
\begin{align}
\bra 0 \bra{\psi_j} A \ket 0 \ket{\psi_j} = -1+2|p|\;.
\end{align}
Since $A$ is unitary, Eq.~\eqref{eq:statement} follows with $\epsilon=2\sqrt{|p|} <1/2$.

To improve the approximation error to $O(\epsilon)$, it suffices to run
$q=O(\log(1/\epsilon))$ PEAs as in Fig.~\ref{fig:pea2}.
The unitary operation $B$ is composed
of $q$ quantum Fourier transforms in parallel and the operation $W$ is
depicted in the figure.
The total number of ancillary qubits
is 
\begin{align}
\label{eq:nPEA}
n=O\left (\log \left(\frac 1 {\epsilon}\right)\log \left(\frac 1{\Delta} \right) \right)
\end{align}
 and $p$ in Eq.~\eqref{eq:PEAaction} is now bounded as, e.g.,  $|p|=O(\epsilon^2)$.
It follows that, for $A=W^\dagger R W$,
\begin{align}
\nonumber
\left \| ( A - \one) \ket 0 \sket{{\psi_0}} \right\| &=0 \;, \\
\left \| ( A + \one) \ket 0 \sket{\psi_j} \right\| &\le \epsilon \;, j > 0 \; ,
\end{align} 
so that Eq.~\eqref{eq:statement} follows and the approximate reflection is implemented.

\begin{figure}
\begin{minipage}[c][]{.47\textwidth}
  \vspace*{\fill}
  \centering
  \subfloat[]{\includegraphics[width=\textwidth]{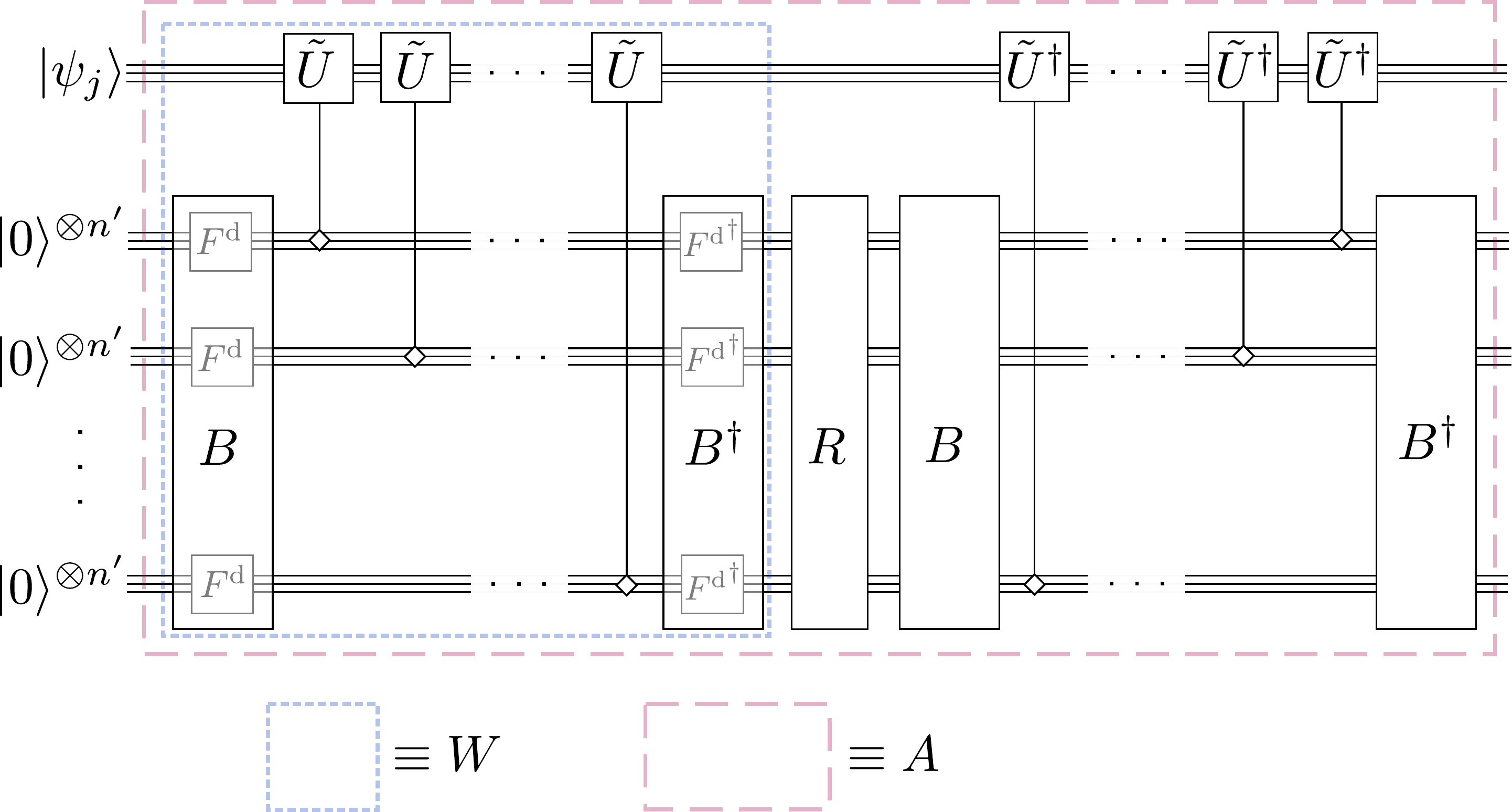}}
  \par
  \subfloat[]{\includegraphics[width=.7\textwidth]{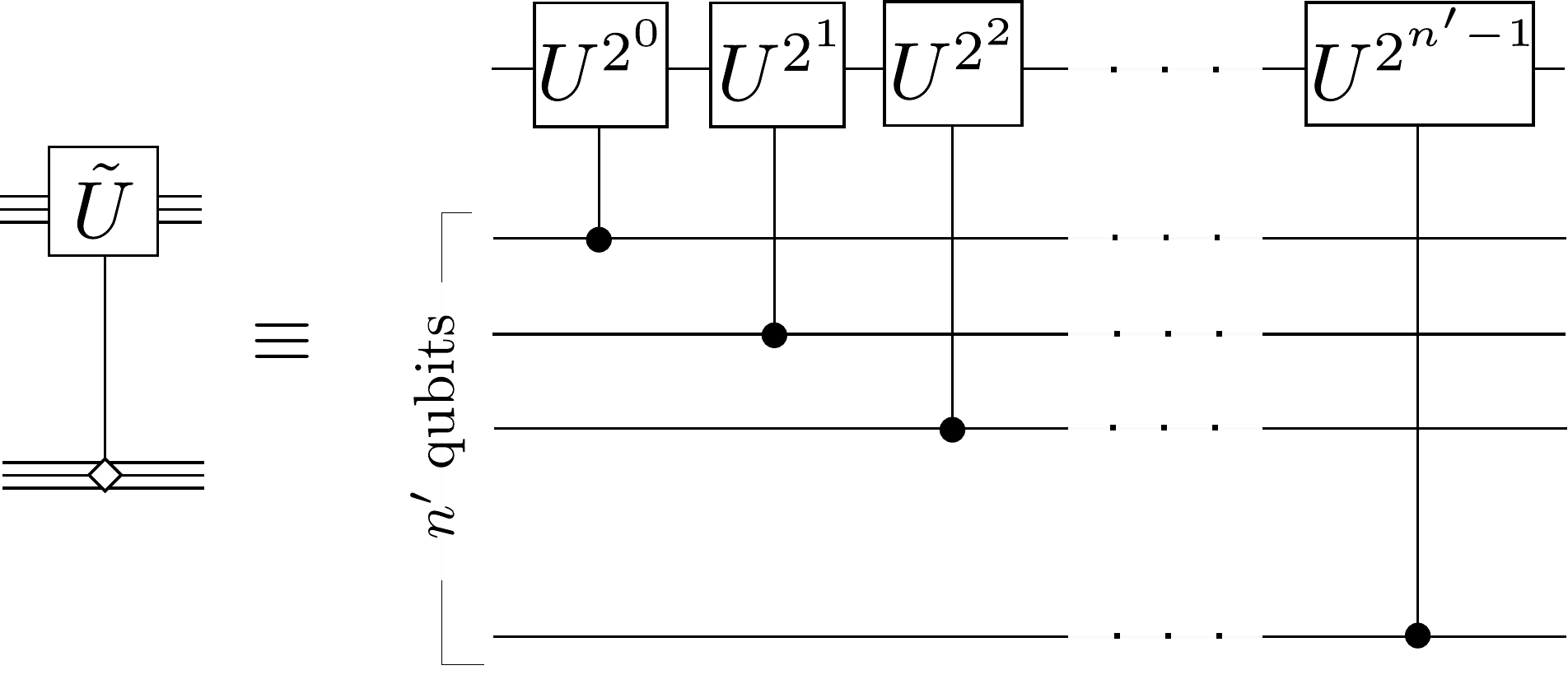}}
\end{minipage}
\caption{(a) Quantum algorithm to implement the approximate reflection $A$ based on 
repeated uses of the PEA. $n'=O(\log(1/\Delta))$. (b) The controlled $\tilde{U}$ operation expanded out.}
\label{fig:pea2}
\end{figure}

\begin{figure}
	\begin{center}
		\includegraphics[width=.22\textwidth]{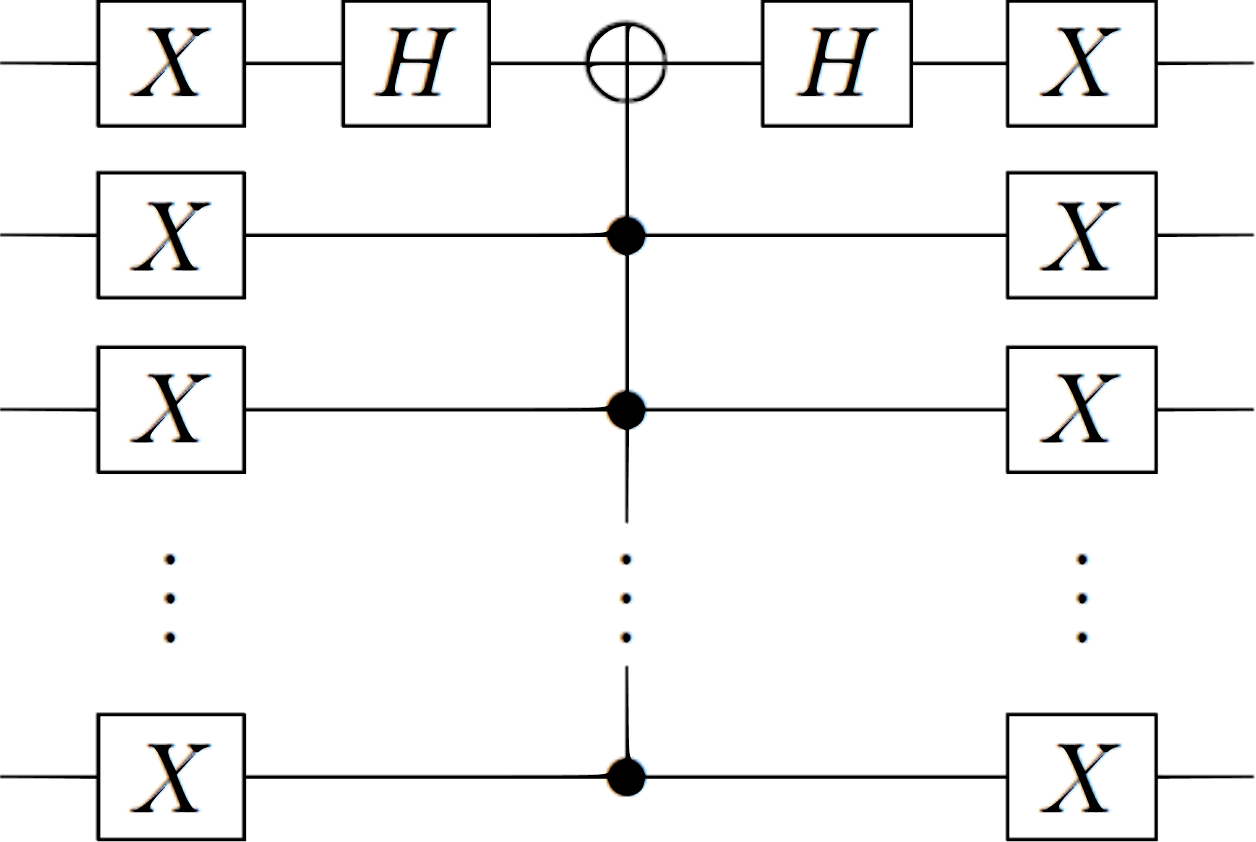}
	\end{center}
	\caption{Circuit to implement the reflection $R$ used in the PEA approach (up to an irrelevant phase of -1). 
	The multiply-controlled NOT gate can be decomposed into a number of two-qubit gates that scales linearly with the number of qubits using
	one additional qubit~\cite{HLZ17}. $X$ and $H$ are the Pauli and Hadamard single-qubit gates, respectively.}
	\label{fig:refancilla}
\end{figure}

The gate complexity of this approach is of order $q$ times
the gate complexity of the quantum Fourier transform, which is
$O(\log^2(1/\Delta))$. Nevertheless, it suffices to implement
each quantum Fourier transform with constant precision
by leaving out some number of controlled rotations with small angles. 
This is because, under the approximation, we still satisfy $A \ket 0 \ket{\psi_0}=\ket 0 \ket{\psi_0}$,
which is a property that has not been exploited in previous works.
Following Ref.~\cite{Cop94}, this reduces the gate complexity of each quantum Fourier transform
to $O(\log(1/\Delta) \log \log(1/\Delta))$. The operation $R$ of Fig.~\ref{fig:pea2} can be implemented by the circuit shown in Fig.~\ref{fig:refancilla} , which requires using two-qubit gates that scale linearly with the size of the ancilla. The overall gate complexity is then
\begin{align}
\label{eq:gatecostPEA}
C_B =   O \left(\log \left(\frac1\epsilon \right )\log \left( \frac 1 \Delta \right) \log\log \left( \frac 1 \Delta \right) \right) \;.
\end{align}
The query complexity is $C_U=O(\log(1/\epsilon)/\Delta)$.

\section{The LCU approach}
\label{sec:LCU}
We will first approximate a reflection operator
by a polynomial in $U$ and $U^\dagger$. We start with the Poisson summation formula, which
states
\begin{align}
\label{eq:Poisson1}
\sum_{k=-\infty}^\infty f(k)= \sqrt{2 \pi}\sum_{l=-\infty}^\infty \hat f(2 \pi l) \;.
\end{align}
Here $f$ is a Schwartz function and $\hat f$ is the (unitary) Fourier transform:
\begin{align}
\hat f(y) = \frac 1 {\sqrt{2\pi}} \int_{-\infty}^\infty dx \;f(x) e^{-i y x} \;,
\end{align}
$y,x \in \mathbb R$. In our case, we choose $f(x)=e^{-((\lambda + 2 \pi x)/\delta z)^2/2}$, for some $\delta z>0$, 
and $\hat f(y)=\delta z \; e^{-(y \delta z)^2/(8\pi^2)} e^{i y \lambda/(2\pi)} /(  {2\pi})$\;.
Then,
\begin{align}
\label{eq:Poisson2}
\sum_{k=-\infty}^\infty e^{-((\lambda + 2 \pi k)/\delta z)^2/2 }= \frac {\delta z} {\sqrt{2 \pi}} \sum_{l=-\infty}^\infty e^{-(l \delta z)^2/2} e^{i l \lambda} \;.
\end{align}
We will relate $\lambda$ to the eigenphase of the unitary $U$.
As expected, the summation formula is invariant under the transformation $\lambda \rightarrow \lambda \pm 2 \pi$.
For the following corollaries, we define
\begin{align}
\alpha_l := \frac {\delta z} {\sqrt{2 \pi}} e^{-(l \delta z)^2/2} \;.
\end{align}

Lemma~\ref{lem:approxerror1} in Appendix~\ref{app:refresults}
implies:
\begin{corollary}
\label{cor:ref1}
Let $1/5 \ge \epsilon >0$ and $\Delta>0$. Then, there exist $\delta z =O(\Delta/\sqrt{\log(1/\epsilon)})$ and $L=O(\log(1/\epsilon)/ \Delta)$
such that
\begin{align}
\left \|\left (  \sum_{l=-L}^{L-1} \alpha_l U^l -\one \right)\ket {\psi_0} \right \| =O( \epsilon) \; 
\end{align}
and, for $0< \Delta \le \lambda_j \le 2 \pi -\Delta$ (i.e., $j>0$), 
\begin{align}
\left \| \left (  \sum_{l=-L}^{L-1} \alpha_l U^l \right )\sket {\psi_j} \right \| =O( \epsilon) \;.
\end{align}
\end{corollary}

The proof follows simply by replacing $U \rightarrow {e^{i\lambda_j}}$
in Lemma~\ref{lem:approxerror1} of Appendix~\ref{app:refresults} and by noticing
that $\delta z e^{-(L \delta z)^2/2}=O(\epsilon)$ if we choose the right constants hidden by the order notation of $L$ and $\delta z$.
Additionally, we obtain:
\begin{corollary}
\label{cor:ref2}
Let $1/5 \ge \epsilon >0$, $\Delta>0$, and $\beta_{-L},\ldots,\beta_{L-1}$ be 
complex numbers such that, for any state $\ket \xi$,
\begin{align}
\label{eq:cor2}
\left \|\left( \sum_{l=-L}^{L-1}(\alpha_l - |\beta_l|/2) U_l \right) \ket \xi \right\| =O(\epsilon) \;.
\end{align}
Then, there exist $\delta z =O(\Delta/\sqrt{\log(1/\epsilon)})$ and $L=O(\log(1/\epsilon)/ \Delta)$
such that
\begin{align}
\label{eq:cor2b}
\left \|\left (\sum_{l=-L}^{L-1} \frac{|\beta_l| } 2 U^{l} -\one \right)\ket {\psi_0} \right \| =O( \epsilon) \; 
\end{align}
%
and, for $0< \Delta \le \lambda_j \le 2 \pi -\Delta$ (i.e., $j>0$), 
\begin{align}
\left \| \left (\sum_{l=-L}^{L-1} \frac{|\beta_l |} 2 U^{l}  \right )\sket {\psi_j} \right \| =O( \epsilon) \;.
\end{align}
\end{corollary}

The proof follows simply from Cor.~\ref{cor:ref1} and the triangle inequality.

Without loss of generality we can choose $L=2^{m-1}$, $m \ge 1$, to be a power of 2.
This will simplify the implementation of certain gates in our quantum algorithm.

\begin{definition}
\label{def:refLCU}
The approximate reflection operator for the LCU approach is
\begin{align}
\tilde R_{\psi_0} =   \sum_{l=-L}^{L-1}| \beta_l| U^{l} -  \one \;.
\end{align}
\end{definition}


That is, $\tilde R_{\psi_0}$ can be written as a polynomial in $U$ and $U^\dagger$,
and approximates $R_{\psi_0}$ since  Corollary~\ref{cor:ref2} and the definition of $\beta_l$ imply
\begin{align}
\label{eq:reflectionerror}
\|R_{\psi_0} - \tilde R_{\psi_0} \| =O( \epsilon) \;.
\end{align}

The important property is that $\tilde R_{\psi_0}$ is a LCU and also approximates a unitary transformation. 
We can then use the results of Ref.~\cite{BCC+15} to build a quantum
algorithm that implements $ R_{\psi_0}$, in the sense of Eq.~\eqref{eq:statement}, as follows.
Without loss of generality, we rewrite
\begin{align}
\label{eq:LCUref}
\tilde R_{\psi_0} = \sum_{l=-L}^{L+2} |\beta_l| \bar U_l \;.
\end{align}
For $-L \le l \le L-1$, the coefficients $\beta_l$ are as in Eq.~\eqref{eq:cor2} and the unitaries are $\bar U_l := U^l$.
For $l=L$ we define $\beta_L:=1$ and $\bar U_L:=-\one$. For $L+1 \le l \le L+2$ the coefficients are
$\beta_l:=(1/\sin \alpha -3)/2$
and the unitaries are $\bar U_{L+1}:=\one$ and $\bar U_{L+2}:=-\one$.
The angle is $\alpha = \pi/10$ and the last two terms in the LCU add up to zero. These terms
are needed to fit the framework of {\em oblivious amplitude amplification} (OAA) introduced in 
Refs.~\cite{BCC+14,BCC+15} as we will explain below.

We assume that there is a mechanism select($\bar U$) to implement controlled-$\bar U$ operations as follows:
\begin{align}
\label{eq:selectU}
{\rm select}(\bar U) \ket l \ket \xi := \ket l \bar U_l \ket \xi \;.
\end{align}
The details of this mechanism are explained in Sec.~\ref{sec:selectU}.
We also assume the existence of a unitary $B$ that acts as
\begin{align}
\label{eq:operationB}
B \ket 0 =\frac 1 {\sqrt s} \sum_{l=-L}^{L+2} \sqrt{\beta_l } \ket l \;.
\end{align}
Here, $s=\sum_{l=-L}^{L+2} |\beta_l|$ and 
$\ket 0$ is the $n$-qubit state $\ket0^{\otimes n}$.
The states $\ket l$ do not necessarily refer to a binary representation
of the integer $l+L$; it  suffices to satisfy $\bra l l' \rangle = \delta_{l,l'}$. The details of $B$ are explained in Sec.~\ref{sec:WB}.
From Eq.~\eqref{eq:cor2b} and~\eqref{eq:LCUref} we obtain
\begin{align}
\label{eq:s1}
|s-1/\sin \alpha| =O (\epsilon )\;.
\end{align}

If we define
\begin{align}
W:= (B^\dagger \otimes \one) ({\rm select}(\bar U))(B \otimes \one) \;,
\end{align} 
then
\begin{align}
\label{eq:Waction2}
W \ket{0} \ket \xi= \frac 1 s \ket 0 \tilde R_{\psi_0} \ket \xi + \sqrt{1 -\frac 1 {s^2}} \ket{\Phi}
\end{align}
for some
$\ket{\Phi}$
whose ancillary state is supported in the
subspace orthogonal to
$\ket 0$. Our goal is to prepare the first term on the right hand side of Eq.~\eqref{eq:Waction2}.

Last, and as in Sec.~\ref{sec:PEA}, we define the $n$-qubit ancilla
reflection operator $R:= 2 P-\one$, where $P :=\ketbra 0 \otimes \one$
is a projector (i.e., $P^2=P$).

Following Refs.~\cite{BCC+14,BCC+15}, if we were to assume that $s$
is exactly $1/\sin (\pi/10)$
and $\tilde R_{\psi_0}$ is an exact unitary operation, we would obtain
\begin{align}
\label{eq:notdefA}
A \ket 0 \ket \xi = \ket 0 \tilde R_{\psi_0} \ket \xi \;,
\end{align}
with
\begin{align}
\label{eq:defA}
A:=W R W^\dagger R W R W^\dagger R W \; ,
\end{align}
also being a unitary operation.
The quantum state on the right hand side of Eq.~\eqref{eq:notdefA} 
is the desired state.
This corresponds to two rounds of OAA rather than one as 
in Refs.~\cite{BCC+14,BCC+15}, reason why 
we chose $\alpha=\pi/10$.

Since neither $s=1/\sin(\pi/10)$ nor  $\tilde R_{\psi_0}$ is a unitary,
our previous assumptions and Eq.~\eqref{eq:notdefA} are invalid. However, due to our error bounds,
it  follows that (Sec.~\ref{sec:LCUcorrect})
\begin{align}
\label{eq:Aapprox}
\|A \ket 0 \ket \xi - \ket 0 \tilde R_{\psi_0} \ket \xi \| =O(\epsilon) \;,
\end{align}
for any $\ket \xi \in \cH$, which is our desired goal.
Then, our quantum algorithm to implement the approximate
reflection is simply the operation $A$.

\subsection{Correctness}
\label{sec:LCUcorrect}
To show that the quantum algorithm for the LCU approach works, we need
to show that Eq.~\eqref{eq:statement} is valid. 
To this end, we note that $PAP$  can be written as,
\begin{align}
PA P=& 5PWP-20PWPW^{\dagger}PWP + \nonumber \\  &+16PWPW^{\dagger}PWPW^{\dagger}PW P \label{eq:Aexp}.
\end{align}
We then use $PWP=(1/s)P\otimes\tilde{R}_{\psi_0}$ and obtain
\begin{align}\label{eq:PAexp}
PA P=P \otimes \Bigg(&\frac{5}{s}\tilde{R}_{\psi_0}-\frac{20}{s^3}\tilde{R}_{\psi_0}\tilde{R}_{\psi_0}^{\dagger}\tilde{R}_{\psi_0}\nonumber \\ &+ \frac{16}{s^5}\tilde{R}_{\psi_0}\tilde{R}_{\psi_0}^{\dagger}\tilde{R}_{\psi_0}\tilde{R}_{\psi_0}^{\dagger}\tilde{R}_{\psi_0} \Bigg) \;.
\end{align}
A simple calculation implies
\begin{align}
|1- 5/s +20/s^3-16/s^5|& =O(\epsilon) \;.
\end{align}
Also, our construction implies
\begin{align}
\| \one - \tilde{R}_{\psi_0}^{\dagger}\tilde{R}_{\psi_0}\|= O(\epsilon)
\end{align}
so that using the triangle inequality
\begin{align}
\| PAP - P \otimes \tilde R_{\psi_0} \| &= O(\epsilon) \;, \\
\| PAP - P \otimes R_{\psi_0} \| &= O(\epsilon) \;.
\end{align}
Since both $A$ and $R_{\psi_0}$ are unitaries, we obtain $\| PAP-AP\|=O(\epsilon)$.
Then Eq.~\eqref{eq:Aapprox} follows from the triangle inequality.
Finally, using the right constants hidden by the order notation 
in the approximation errors (see Lemma~\ref{lem:approxerror1}), Eq.~\eqref{eq:reflectionerror} implies Eq.~\eqref{eq:statement}.

\subsection{The operation $B$}
\label{sec:WB}
Our quantum algorithm uses the operation $B$
defined in Eq.~\eqref{eq:operationB}. 
In Appendix~\ref{app:gaussianprep} we prove the existence of a quantum algorithm $\hat B$
that prepares the quantum state
\begin{align}
\frac 1 {\sqrt 2} \sum_{l=-L}^{L-1} \sqrt{\beta_l} \ket l \;,
\end{align}
where $\ket{-L} =\ket {0 \ldots 00}, \ket{-L+1}=\ket{0 \ldots 01}, \ket{L-1}=\ket{1\ldots11}$ and the number of qubits
is $m=\log_2(2L)$. The parameter $L$ is as in Lemma~\ref{lem:approxerror1} and Corollary~\ref{cor:ref1}.
The complex numbers $\beta_l$ satisfy Eq.~\eqref{eq:cor2}.

The quantum state of Eq.~\eqref{eq:operationB} has $2L+3$ amplitudes. To build $B$ we start
with a quantum algorithm that prepares the two-qubit ancillary state proportional to
\begin{align}
\label{eq:manyqubitstate}
\sqrt 2 \ket {00} +\sqrt{\beta_{L}} \ket{01} + \sqrt{\beta_{L+1}} \ket{10} + \sqrt{\beta_{L+2} } \ket {11} \;.
\end{align}
This can be done with constant gate complexity.
We then add a system of $m$ qubits initialized in $\ket 0^{\otimes m}$. Last, we apply $\hat B$ conditional
on the first two qubits being in $\ket{00}$. The prepared state is
\begin{align}
\label{eq:manyqubitstate2}
\ket{00} \sum_{l=-L}^{L-1} \sqrt{\beta_l} \ket l + \sqrt{\beta_{L}} \ket{010 \ldots 0} + \\
\nonumber
+ \sqrt{\beta_{L+1}} \ket{10 0 \ldots 0} + \sqrt{\beta_{L+2} } \ket {110 \ldots 0} \;,
\end{align}
where $l$ denotes the integer $l+L$ in binary using $m$ bits.

\subsection{The operation ${\rm select}(\bar U)$}
\label{sec:selectU}
The operation ${\rm select}(\bar U)$ acts as in Eq.~\eqref{eq:selectU}.
 We label the $n$ qubits of the ancillary state
as $1,2,\ldots,n=2+m$. The first two qubits are the ancillary qubits used to prepare
the state of Eq.~\eqref{eq:manyqubitstate}. Conditional on the state of
these two qubits being $\ket{01}$, $\ket{10}$, and $\ket{11}$, ${\rm select}(\bar U)$
applies the unitary $-\one$, $\one$, and $-\one$, respectively. This operation can be summarized
with the diagonal Pauli operator $Z$ acting on the second qubit. Next, conditional 
on the state of these first two qubits being in $\ket{00}$,
${\rm select}(\bar U)$ applies the operation $U^{-L}$. Last,
conditional on the state of the $n$ qubits being in $\ket{00b_{3} \ldots b_n}$ ($b_i \in \{0,1\}$),
${\rm select}(\bar U)$ applies the operation
\begin{align}
 U^{\sum_{j=3}^n 2^{n-j} b_j} \;.
\end{align}
An example of a ${\rm select}(\bar U)$ operation is shown in Fig.~\ref{fig:selectU}. Its action 
for different basis states of the ancillary system is 
\begin{figure}
\begin{center}
\includegraphics[width=.35\textwidth]{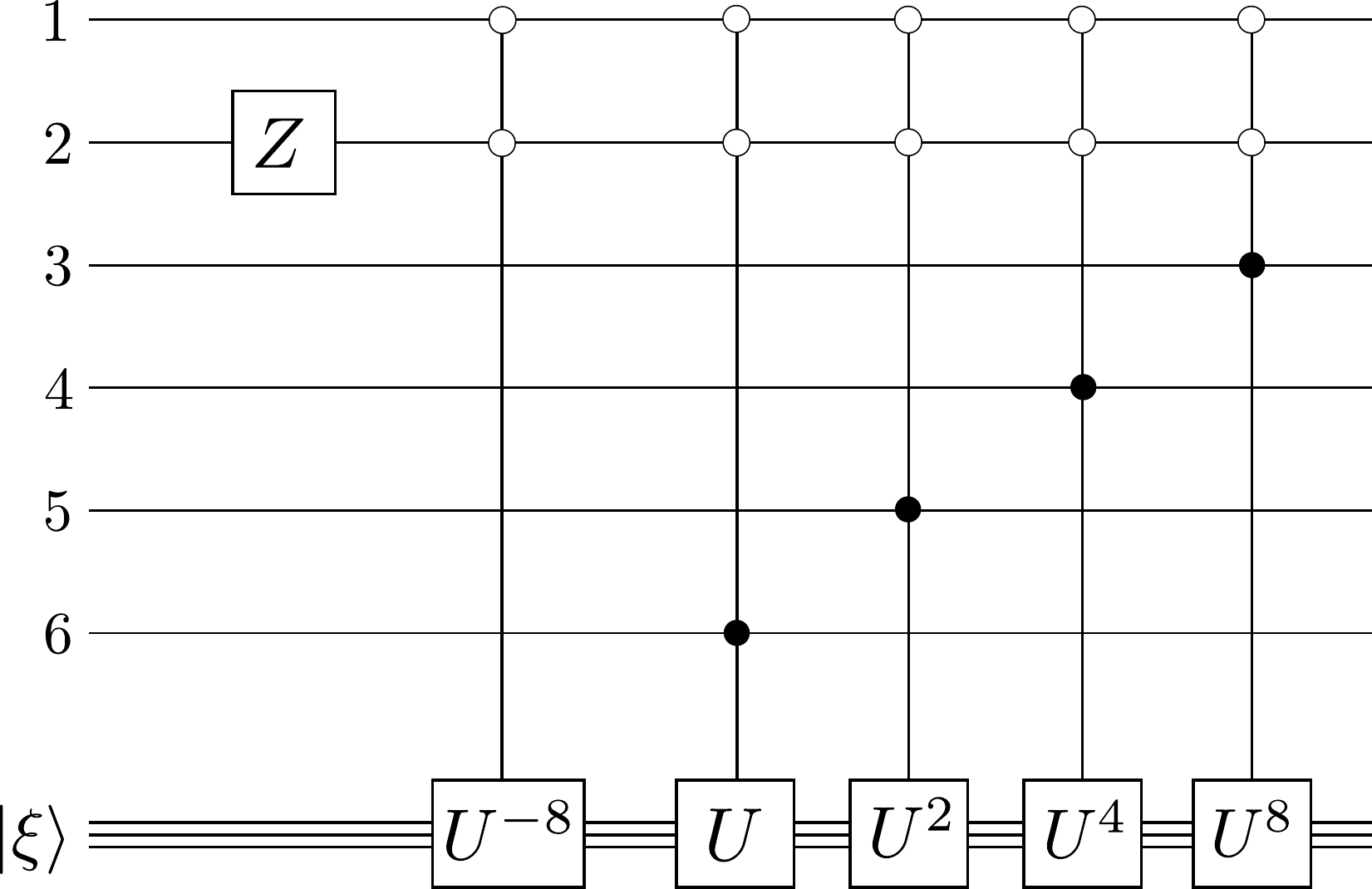}
\end{center}
\caption{The ${\rm select}(\bar U)$ operation for $n=6$ and $L=8$.}
\label{fig:selectU}
\end{figure}

\begin{align}
{\rm select}(\bar U)
\left \{ 
\begin{matrix*}[l]
\ket{00} \ket l \ket \xi & \rightarrow &\ket{00} \ket l U^{l-8} \ket \xi \;, & 0 \le l \le 15 \;, \cr
\ket{010000}  \ket \xi & \rightarrow &-\ket{010000} \ket \xi \;, \cr
\ket{100000}  \ket \xi & \rightarrow &\ket{100000} \ket \xi \;, \cr
\ket{110000}  \ket \xi & \rightarrow &-\ket{110000} \ket \xi \;.
\end{matrix*}
\right .
\label{eq:selectUb}
\end{align}

\subsection{Ancillary qubits and complexity}
\label{sec:LCUcomplex}
The operation $A$ uses the operation  $B$ of Sec.~\ref{sec:WB}.
Then, the number
of ancillary qubits  is $n=2+m$.
Using the results of Lemma~\ref{lem:approxerror1}, we obtain
\begin{align}
\label{eq:nLCU}
 n=O\left(  \log  \log(1/\epsilon) +  \log(1/\Delta) \right) \;.
 \end{align}
 The number of ancillary qubits    for the LCU approach
is then a significant improvement with respect to the   number of ancillary qubits of the PEA approach --
see Eq.~\eqref{eq:nPEA}   for a comparison.

The operation $A$ uses $W$ and $W^\dagger$
five times. It follows that $ B$ ($\hat B$) and $ B^\dagger$ [$(\hat B)^\dagger$]
are also used a constant number of times. 
Following Sec.~\ref{sec:WB} and Appendix~\ref{app:gaussianprep},
\begin{align}
\hat B \ket 0 = \tilde F^{\rm d}_{\rm c} \ket \phi \;,
\end{align}
where  $\tilde F^{\rm d}_{\rm c}$ is an $O(\epsilon)$ approximation of the centered Fourier transform $  F^{\rm d}_{\rm c}$.
According to Eq.~\eqref{eq:fcddecomp}, the centered Fourier transform
uses the QFT three times. Thus, its gate complexity is of the same order as
that of the QFT. The approximate centered Fourier transform $\tilde F^{\rm d}_{\rm c}$
uses an approximate unitary QFT that is obtained by avoiding those phase gates
where the phases are sufficiently small. Using the results of Ref.~\cite{Cop94}, the gate complexity of $\tilde F^{\rm d}_{\rm c}$
is
\begin{align}
O \left( m \log(m/\epsilon)\right).
\end{align}

Lemma~\ref{lem:stateprep1} in Appendix~\ref{app:gaussianprep} 
implies that the state $\ket \phi$ is a superposition of $2L^* = O(\log(1/\epsilon))$
basis states.
Then, Ref.~\cite{SBM06} provides a method
to prepare such a state over $\log_2(2 L^*)$ qubits with gate complexity that is $O(L^*)$.
It is 
important to remark that no other ancillary qubits are needed
to prepare $\ket{\phi}$. The method in Ref.~\cite{SBM06}
requires precomputing $O(L^*)$ rotation angles classically
with sufficiently high precision.
This results in an additional classical complexity that we do 
not consider here since this step has to be done only one time
and does not change the quantum gate complexity.

The state thus prepared is only on a register of $\log_2 (2L^*)$ qubits however, whereas the $\tilde F^{\rm d}_{\rm c}$ operation acts on a Hilbert space of $m=\log_2 (2L)$ qubits. Specifically, $\tilde F^{\rm d}_{\rm c}$ needs to act on a state of the form $\sum_{l=-L}^{L-1}\gamma_l\ket{l}$ where $\gamma_l$'s correspond to the amplitudes of $\ket{\phi}$ for $-L^* \leq l \leq L^*-1$ and are 0 otherwise. In other words, the state $\ket{\phi}$ needs to be centered on a register of  $m$ qubits that encodes $2L$ basis states. This can be done using $ \log_2(L-L^*) = O\left(\log(\log(1/\epsilon)/\Delta)\right)$ two qubit gates, as shown in Appendix~\ref{app:centering}.

In the limit where $\Delta = O\left(\frac{1}{\log(1/\epsilon)}\right)$, which includes cases where $\Delta \ll 1$, $\epsilon \ll 1$, we obtain $m=O(\log(1/\Delta))$.
The overall gate complexity of $B$ in this limit is then
\begin{align}
C_B=O \left( \log \left(\frac 1 \Delta \right) \log \left(  \frac{\log \left(\frac 1 \Delta \right)}{\epsilon}  \right)\right) \;.
\end{align}
This is comparable to the gate complexity obtained in the PEA approach -- see Eq.~\eqref{eq:gatecostPEA}.

The query complexity to implement ${\rm select}(\bar U)$ is, at most, $L$.
So the total query complexity of the LCU approach is
$C_U = O(\log(1/\epsilon)/\Delta)$.
This is similar to the query complexity of the PEA approach.


\section{Reflections and Hamiltonians}
\label{sec:Hamilt}
In this section we discuss the case of
 reflections over eigenstates of Hamiltonians. 
 This case is relevant for, e.g., Ref.~\cite{BKS10}.
 We let $H$ be a Hamiltonian acting on states in $\cH$ such that
\begin{align}
H \ket {\psi_j} = \lambda_j \ket {\psi_j} \; ,
\end{align}
$j=0,1,\ldots,D-1$ and $\ket {\psi_0}$ is the target state.
The eigenvalues satisfy
\begin{align}
\label{eq:eigenvprop}
|\lambda_0 - \lambda_j| \ge \Delta \ , j > 0 \;.
\end{align}
That is, $\Delta$ is a lower bound on the spectral gap. Since we work with finite dimensional 
Hilbert spaces, we can assume $\|H \| \le 1$.

We seek an approximation of the operator that makes reflections
over $\ket {\psi_0}$. Then, under the assumptions, we can readily use 
the results for the unitary case if we take
\begin{align}
U := e^{i(H-\lambda_0)} \;.
\end{align}
This is the evolution operator 
induced by $H$ for a unit of time. 

We consider the scenario described
in Sec.~\ref{sec:problem} where the matrix elements of $H$
can be queried. References~\cite{BCC+14,BCC+15,LC17}
provide then a way to construct an approximation of $U$
using the queries ${\cal Q}_H$.

\begin{definition}
The query complexity of implementing an approximate reflection
over $\ket {\psi_0}$ in the Hamiltonian case, $C_H$, is the number
of times the procedure $\cQ_H$ is invoked.
\end{definition}

Clearly, $C_H$ will depend on the approximation error
and the Hamiltonian simulation method used to implement an approximation of $U$.
For example, using the method in Ref.~\cite{LC17}, we can construct a unitary operator $U'$  
using the procedure $\cQ_H$
\begin{align}
 O \left(d + \frac {\log(1/\varepsilon)}{\log\log(1/\varepsilon)}\right) \;
\end{align}
times, and 
\begin{align}
\|U' \ket 0^{\otimes n'_H} \ket \xi- \ket 0^{\otimes n'_H} U \ket \xi \| \le \varepsilon \;. 
\end{align}
$n'_H$ is the number of ancillary qubits required to implement $U'$.
We note that in the approximation of $R_{\psi_0}$ of Def.~\eqref{def:refLCU}, the degree of the polynomial
in $U$ and $U^\dagger$ is $L=O(\log(1/\epsilon)/\Delta)$. We will then choose
$\varepsilon = O(\epsilon/L)$ and define:
\begin{definition}
\label{def:refHamil}
The approximate reflection operator in the Hamiltonian case is
\begin{align}
\tilde R_{\psi_0} =   \sum_{l=-L}^{L-1}| \beta_l| U'^{l} -  \one  \;.
\end{align}
\end{definition}

This definition implies
\begin{align}
\|\tilde R_{\psi_0}\ket 0^{\otimes n'_H} \ket \xi- \ket 0^{\otimes n'_H} R_{\psi_0} \ket \xi \| =O( \epsilon) \;.
\end{align}

Following Sec.~\ref{sec:LCU}, the quantum algorithm to implement the approximate reflection is then
\begin{align}
A:= WRW^\dagger RWRW^\dagger RW \;, 
\end{align}
and
\begin{align}
W:= (B^\dagger \otimes \one) ({\rm select} (\bar U)) (B \otimes \one) \;.
\end{align}
The operation $B$ is the one described in Sec.~\ref{sec:WB}. The operation ${\rm select} (\bar U)$
is similar to the one described in Sec.~\ref{sec:selectU} with the only difference being that the unitary $U$ is replaced by the unitary $U'$. 
$R$ is a reflection operator acting on $n$ qubits as in Eq.~\eqref{eq:nLCU}.
Using the right constants in the order notation, this definition of $A$ implies Eq.~\eqref{eq:statement}.

 It follows from Ref.~\cite{LC17} that the query complexity of implementing $A$ is
 \begin{align}
 C_H = O \left(  L \left(d + \frac {\log(L/\epsilon)}{\log\log(L/\epsilon)} \right) \right) \;,
 \end{align}
 where $L=O(\log(1/\epsilon)/\Delta)$ has been determined in  Lemma~\ref{lem:approxerror1}.
 The number of additional two-qubit gates
 also depends on the Hamiltonian simulation method
 that is used to implement $A$.  For example, following Ref.~\cite{LC17},
 the gate complexity is dominated by that of the Hamiltonian simulation method and is
 \begin{align}
C_B= O\left( (\log D + h \; {\rm polylog} (h)  )C_H\right) \;,
 \end{align}
 where $h$ is the number of bits of precision of the matrix elements of $H$.
 Last, the total number of ancillary qubits resulting from Ref.~\cite{LC17} for the Hamiltonian case is also
 \begin{align}
 n_H = O \left(  \log \log(1/\epsilon)+ \log(1/\Delta)  \right) \;,
 \end{align}
i.e., $n'_H=O(1)$.

 As in the LCU approach, the total number of qubits is an improvement
 with respect to those needed if we followed the PEA approach
 for the current case.
 

%

\section{Lower bound on query complexity}
\label{sec:lowerbound}
In this section we obtain a lower bound on the query complexity of performing a reflection over an eigenvector of the unitary operator
$U$. The proof is based on the optimality of Grover's search algorithm. 
We consider the unstructured search problem with a unique marked element $\ket{t}$ in a search space of size $D$
and write $\ket s$ for the equal superposition state. In Ref.~\cite{BCWZ99} it was shown that the number of 
queries to the black box needed to solve this problem with a quantum computer
and probability greater or equal than $1-\nu$ is $\Theta (\sqrt{D \log(1/\nu)})$, with $\nu \ge 2^{-D}$.

We define
\begin{align}
	\sket{\tilde \psi_0} &:= \frac{\ket{s} + \ket{t}}{\sqrt{2(1+1/\sqrt{D})}}\, .
\end{align}
and note that
\begin{align}
	R_{\tilde \psi_0} \ket{s} = \ket{t} \;,
\end{align}
with $R_{\tilde \psi_0} := 2 \ketbra{\tilde \psi_0} - \one$ being also a reflection operator.
It follows that the search problem can be solved exactly
with a single application of $R_{\tilde \psi_0}$

Motivated by the action of $R_{\tilde \psi_0}$, we let $R_t:=2 \ketbra t - \one$ and $R_s:=2 \ketbra s - \one$ be reflection operators over $\ket{t}$ and $\ket{s}$, respectively. 
In Grover's search algorithm, $R_t$ is implemented with a single query to the black box.
We further define the following unitary operators 
\begin{align}
	V &:= e^{i \frac{\pi}{2} \ket{s}\bra{s}} = \one + (i-1) \ket{s}\bra{s} \, , \\
	U &:= -e^{-i \cos^{-1}\left( 1-\frac{2}{D} \right)} V^\dagger R_s R_t V
\end{align}
$U$ has a unique eigenvector of eigenvalue 1, which we denote $\sket{ \psi_0}$,
and approximates $\sket{\tilde \psi_0}$ in the limit of large $D$.
The other eigenvalues are such that $\Delta = O(1/\sqrt D)$.

The reflection operator over $\ket{\psi_0}$ can be shown to satisfy
\begin{align}
\label{RSzero}
	\bra{s} R_{\psi_0}\ket{s} = 0 \, .	
\end{align}
Additionally, we let $\tilde R_{\psi_0}$ be the approximate reflection
that satisfies
\begin{align}
\label{RRepsilon}
	\| \tilde R_{\psi_0} - R_{\psi_0}\| \le \epsilon \;.
\end{align}
The approximate reflection operator can be used to solve 
the unstructured search problem by acting on $\ket s$ with failure probability $\nu = 1- | \bra{t}\tilde{R}_{\psi_0}\ket{s} |^2$. 
Using Eqs.~\eqref{RSzero} and~\eqref{RRepsilon}, and the fact that $\braket{t}{s}=1/\sqrt{D}$, 
this failure probability can be upper bounded by
\begin{align}
	\nu = O\left(\left(\frac{1}{\sqrt{D}}+\epsilon\right)^2\right) \, .
\end{align}
Moreover, we can always choose $\epsilon=O(1/\sqrt{D})$ such that $\nu \ge 2^{-D}$. 
Then the results of Ref.~\cite{BCWZ99} can be applied to obtain a lower bound on the query complexity. 
Since $U$ makes a single query to the black box, it follows
that the query complexity of $\tilde R_{\psi_0}$ is
\begin{align}
\Omega(\sqrt{D \log(1/(1/\sqrt{D}+\epsilon))}) \;.
\end{align}
In terms of the eigenphase gap $\Delta$, this is
\begin{align}
\Omega((1/\Delta)\sqrt{\log(1/(\Delta+\epsilon))}) \; ,
\end{align}
which is valid for $\epsilon = O(\Delta)$.

\section{Acknowledgements}
ANC thanks G.\ Muraleedharan and N.\ Wiebe for helpful discussions, and G.H. Low for pointing out
Ref.~\cite{Tul16}. ANC was supported by a Google Research Award through a part of the duration of this project.
RS and YS were supported by the LDRD program at Los Alamos National Laboratory.

\appendix
\section{Approximate reflections}
\label{app:refresults}
Our definition for the approximate reflection 
operator in the LCU approach follows the results
of corollaries~\ref{cor:ref1} and~\ref{cor:ref2}. These corollaries are
 a consequence of the following lemma:
\begin{lemma}
\label{lem:approxerror1}
Let $1/5 \ge \epsilon>0$ and $\Delta>0$. Then, there exist $\delta z =O(\Delta/\sqrt{\log(1/\epsilon)})$ 
and $L=O(\log(1/\epsilon)/ \Delta)$ such that
\begin{align}
\label{eq:PoissonIneq1}
\left |\frac {\delta z} {\sqrt{2 \pi}} \sum_{l=-L}^{L-1} e^{-(l \delta z)^2/2} e^{i l \lambda} -1 \right | =O( \epsilon)
\end{align}
if $\lambda =0$ and 
\begin{align}
\label{eq:PoissonIneq2}
\left |\frac {\delta z} {\sqrt{2 \pi}} \sum_{l=-L}^{L-1} e^{-(l \delta z)^2/2} e^{i l \lambda} \right | =O( \epsilon)
\end{align}
if $\Delta\le \lambda \le 2 \pi-\Delta$.
\end{lemma}
\begin{proof}
To prove Eq.~\eqref{eq:PoissonIneq1}, we will show first that the terms with $k \ne 0$ are $O(\epsilon)$
in the left hand side of Eq.~\eqref{eq:Poisson2} with the proper choice of $\delta z$ and for $\lambda=0$. 
First, we assume that $L = \infty$ so we can use the Poisson formula.
If $\delta z \le \pi /\sqrt{\log(c/\epsilon)}$, for some constant $c>1$, we obtain
\begin{align}
\nonumber
\sum_{k \ne 0} e^{-( 2 \pi k /\delta z)^2/2 }& \le \sum_{k \ne 0} (\epsilon/c)^{2 k^2}\\
\nonumber
& \le   2 {\epsilon^2}/({c^2-\epsilon^2})\\
& \le \epsilon/2c \;,
\end{align}
where we used $\epsilon \le 1/5$.
To prove the case of $\lambda \ne 0$, we need to show that all terms in the sum of the left hand side of Eq.~\eqref{eq:Poisson2} are small.
We note that this sum is invariant under the transformation $\lambda \rightarrow \lambda+ 2 \pi$
so we can assume that $\lambda \in [-\pi,-\Delta] \cup [\Delta,\pi)$.
We assume first that $\pi \ge \lambda \ge \Delta$ and the other case can be analyzed similarly.
The term with $k=0$ is $e^{-(\lambda/\delta z)^2/2}$.  This is smaller than $\epsilon/(4c)$ if we choose
$\delta z \le  \Delta/\sqrt{2 \log(4c/\epsilon)}$. Additionally,
\begin{align}
\sum_{k \ne 0} e^{-( (\lambda +2 \pi k) /\delta z)^2/2 } \le  \sum_{k \ne 0} e^{-(  \pi k /\delta z)^2/2 } \;.
\end{align}
As in the previous case, we can make this term smaller than $\epsilon/(4c)$ by choosing
$\delta z \le \pi/\sqrt{\log(2c/\epsilon)}$. Therefore, there is a $\delta z=O(\Delta/\sqrt{\log(1/\epsilon)})$ such that the
right hand sides of Eqs.~\eqref{eq:PoissonIneq1}
and~\eqref{eq:PoissonIneq2}
are bounded by $\epsilon/(2c)$, in the limit $L=\infty$.

To conclude the proof, we analyze the terms
in the Poisson summation formula with $|l|\ge L$, for some $L<\infty$ determined below.
We note
\begin{align}
\nonumber
\left |\frac { \delta z} {\sqrt{2 \pi}} \sum_{\substack{l\ge L \\ l<-L}} e^{-(l \delta z)^2/2} e^{-i l \lambda} \right |& \le \frac {\delta z} {\sqrt{2 \pi}} \sum_{|l|\ge L} e^{-(l \delta z)^2/2} \\
\nonumber
& \le \frac 2 {\sqrt{2 \pi}} \int_{x=(L-1) \delta z}^\infty dx \; e^{-x^2/2} \\
\label{eq:ChernoffIneq}
& \le 2 e^{-((L-1) \delta z)^2/2} \; ,
\end{align}
where the last inequality follows from the Chernoff bound~\cite{Pro95}. Then, this term can be made at most $\epsilon/(2c)$
if $(L-1) \delta z \ge \sqrt{2 \log(4c/\epsilon)}$.
It follows that there exists $L =O(\log{(1/\epsilon)}/\Delta)$
such that Eq.~\eqref{eq:ChernoffIneq} is upper bounded by $\epsilon/(2c)$.
Using the triangle inequality we conclude the proofs of Eqs.~\eqref{eq:PoissonIneq1} and ~\eqref{eq:PoissonIneq2}.
We can choose the constant $c$   to obtain exact bounds hidden by the order notation.
\end{proof}

\section{Preparation of states with Gaussian like amplitudes}
\label{app:gaussianprep}
We seek a quantum algorithm $\hat B$ that prepares an approximation of the (unnormalized) state
\begin{align}
\ket \psi = \left ( \frac {\delta z}{\sqrt{2\pi}}\right)^{1/2} \sum_{l=-L}^{L-1} e^{-(l \delta z)^2/4} \ket l \;,
\end{align}
where the states are ordered in the computational basis such that $\ket {-L} = \ket{0 \ldots 00}, \ket{-L+1}=\ket{0 \ldots 01}, \ldots \ket{L-1}=\ket{1\ldots 11}$.
$L$ is as in Lemma~\ref{lem:approxerror1} and, without loss of generality, $L$ is a power of 2: $L=2^{m-1}$. 

It will be useful to introduce the ``centered '' Fourier transform $F^{\rm d}_{\rm c}$:
\begin{align}
F^{\rm d}_{\rm c} = X^{L} . F^{\rm d}. X^{L} \;,
\end{align}
where $F^{\rm d}$ is the standard quantum Fourier transform of dimension $2L$ (i.e., acting on $m$ qubits)
and $X$ is the cyclic permutation
\begin{align}
X = \begin{pmatrix}  0 & 1 & 0 & \cdots & 0 \cr 0 & 0 & 1 & \cdots & 0 \cr \vdots & \vdots & \vdots & \cdots & \vdots \cr 1 & 0 & 0 & \cdots & 0 \end{pmatrix} \;.
\end{align}
We note that $X=F^{\rm d} .Z . (F^{\rm d})^{-1}$, where $Z$
is the diagonal operation that has the roots of unity as diagonal entries.
In particular,
\begin{align}
Z^{L} = \begin{pmatrix} 1 & 0 & 0 & \cdots & 0 \cr 0 & -1 & 0 & \cdots & 0 \cr 0 & 0 & 1 & \cdots & 0 \cr \vdots & \vdots & \vdots & \cdots & \vdots \cr
0 & 0 & 0 & \cdots & -1\end{pmatrix} \;.
\end{align}
If we label the qubits as $0,1,\ldots,m-1$, $Z^L$ is equivalent to the action of the diagonal Pauli operator $\sigma_z^0$.
Then,
\begin{align}
\label{eq:fcddecomp}
F^{\rm d}_{\rm c} =F^{\rm d}. \sigma_z^0 . F^{\rm d}. \sigma_z^0 . (F^{\rm d})^{-1}  \;.
\end{align}

\begin{lemma}
\label{lem:stateprep1}
Let $L$, $\delta z$, $\epsilon$, and $\Delta$ be as in Lemma~\ref{lem:approxerror1}. Then, there exists $L^* =O(\log(1/\epsilon))$  such that
\begin{align}
\label{eq:stateprep1}
\| \ket \psi - F^{\rm d}_{\rm c} \ket{\phi}  \| 
 = O(\epsilon) \;,
\end{align}
where 
\begin{align}
\ket{\phi}=\frac 1 {\sqrt{\cal N}} \sum_{l=-L^*}^{L^*-1} e^{-(l \pi/(L \delta z))^2} \ket l 
\end{align}
and ${\cal N} = \sum_{l=-L^*}^{L^*-1} e^{-2(l \pi/(L \delta z))^2}$.
\end{lemma}
\begin{proof}
We let $T=2 \sqrt{\pi L}$ be a variable that refers to a period and $\gamma = \sqrt{\pi/L}$ be a variable
that refers to a size of a discretization.
We  define the following (unnormalized) states:
\begin{align}
\ket{\phi'}& :=\left( \frac{\sqrt{2\pi}}{L \delta z}\right)^{1/2} \sum_{l=-L}^{L-1} c_l \ket l \;,\\
\ket{\psi'}& := \left( \frac{\delta z} {\sqrt{2 \pi}}\right)^{1/2}\sum_{l=-L}^{L-1} d_l \ket l \;,
\end{align}
where the amplitudes are
\begin{align}
c_l = \sum_{k=-\infty}^\infty e^{-(l \gamma + kT)^2 \pi/(L \delta z^2)} \;, \\
d_l = \sum_{k=-\infty}^\infty e^{-(l \gamma + kT)^2 L \delta z^2/(4 \pi)} \;.
\end{align}
Following Ref.~\cite{Som15}, it can be shown that
\begin{align}
\ket{\psi'} = F^{\rm d}_{\rm c} \ket{\phi'} \;.
\end{align}

We will use the triangle inequality to prove Eq.~\eqref{eq:stateprep1}.
We note that
\begin{align}
\label{eq:approx1}
\| \ket \psi - \ket {\psi'} \|^2 =  \frac{\delta z} {\sqrt{2 \pi}} \sum_{l=-L}^{L-1} \left | \sum_{k \ne 0} e^{-(l \gamma + kT)^2 L \delta z^2/(4 \pi)} \right |^2 \;.
\end{align}
Also, $|l \gamma + kT| \ge k T/2$ so the right hand side of Eq.~\eqref{eq:approx1}
can be bounded by
\begin{align}
 \frac{2\sqrt 2 L \delta z} {\sqrt{ \pi}} \sum_{k>1} e^{-k^2 L^2 \delta z^2/4}
& = O \left ( \sqrt{\log(1/\epsilon)} \epsilon^3 \right) \\
\nonumber
& = O(\epsilon^2)
\end{align}
if we choose $L$ and $\delta z$ such that $L \delta z \ge \sqrt{12 \log(1/\epsilon)}$.

We also define the state
\begin{align}
\ket{\phi''}:=\left( \frac{\sqrt{2\pi}}{L \delta z}\right)^{1/2} \sum_{l=-L}^{L-1} e^{-(l \gamma)^2 \pi/(L \delta z^2)} \ket l \;.
\end{align}
Then,
\begin{align}
\label{eq:approx2}
\| \ket{\phi'} -\ket{\phi''} \|^2 =  \frac{\sqrt{2\pi}}{L \delta z}\sum_{l=-L}^{L-1} \left| \sum_{k \ne 0} e^{-(l \gamma + kT)^2 \pi/(L \delta z^2)} \right|^2 \;.
\end{align}
The right hand side of Eq.~\eqref{eq:approx2}
can be bounded as
\begin{align}
 \frac{4 \sqrt{2\pi}}{ \delta z} \sum_{k>1} e^{-k^2 \pi^2/\delta z^2}
&= O\left(\frac{\sqrt{\log(1/\epsilon)}}{\Delta} \epsilon^{\pi^2/\Delta^2} \right) \\
\nonumber & = O \left( \epsilon^2 \right) \;
\end{align}
if $\delta z \le \Delta/ \sqrt{\log(1/\epsilon)}$. To obtain the correct order we used $\sqrt{\log(1/\epsilon)}/\Delta \le (1/\epsilon)^{2/\Delta^2}$
for $\epsilon \le 1/5$.

For some $L^*\ge 1$ that we choose below, we now let 
\begin{align}
\ket{\phi'''}:=\left( \frac{\sqrt{2\pi}}{L \delta z}\right)^{1/2} \sum_{l=-L^*}^{L^*-1} e^{-(l \gamma)^2 \pi/(L \delta z^2)} \ket l \;.
\end{align}
Then,
\begin{align}
\label{eq:approx3}
\| \ket{\phi''} - \ket{\phi'''}\|^2 =  \frac{\sqrt{2\pi}}{L \delta z} \sum_{\substack{l\ge L^* \\ l <-L^*}} e^{-2(l \gamma)^2 \pi/(L \delta z^2)} \;.
\end{align}
The right hand side of Eq.~\eqref{eq:approx3} is
\begin{align}
O \left(  e^{-2((L^*-1) \pi/(L \delta z))^2}\right) \;.
\end{align}
According to Lemma~\ref{lem:approxerror1}, the parameters $L$ and $\delta z$ satisfy $L^2 \delta z^2 = c \log(1/\epsilon)$, for some constant $c>0$.
We can then choose $L^* = O(\log(1/\epsilon))$ such that
the right hand side of Eq.~\eqref{eq:approx3} is
\begin{align}
O ( \epsilon^2) \;.
\end{align}

The states $\ket \phi$ and $\ket{\phi'''}$ are proportional to each other. Since $\| \ket \phi \|=1$, we obtain
\begin{align}
\| \ket \phi - \ket{\phi'''} \|^2 = \left| 1 - \| \ket{\phi'''}\| \right|^2 \;.
\end{align}
Also,
\begin{align}
\label{eq:approx4}
 \| \ket{\phi'''}\| ^2 = \frac{\sqrt{2\pi}}{L \delta z} \sum_{l=-\infty}^{\infty} e^{-2(l \gamma)^2 \pi/(L \delta z^2)} + O(\epsilon^2)\;.
\end{align}
Using the Poisson summation formula, the first term on the right hand side of Eq.~\eqref{eq:approx4} is
\begin{align}
\sum_{l=-\infty}^{\infty} e^{-(l L \delta z)^2/2} \;.
\end{align}
This is also
\begin{align}
1 + O(\epsilon^2)
\end{align}
if $L$ and $\delta z$ are chosen such that $L^2 \delta z^2 \ge 4 \log(1/\epsilon)$. It follows that
\begin{align}
\| \ket \phi - \ket{\phi'''} \|^2 = O(\epsilon^2) \;.
\end{align}

Finally, using the triangle inequality, we obtain
\begin{align}
\nonumber
\| \ket \psi - F^{\rm d}_{\rm c} \ket{\phi}  \|  \le & \| \ket \psi - \ket{\psi'}\| + \| \ket{\phi'} -\ket{\phi''}\| + \\
\nonumber & \| \ket{\phi''} -\ket{\phi'''}\| + \| \ket{\phi'''} -\ket{\phi}\| \\
& = O(\epsilon)\;.
\end{align}
This proves the Lemma. 
\end{proof}

We are now ready to prove Eq.~\eqref{eq:cor2}. 
We let $\tilde F^{\rm d}_{\rm c}$ be a unitary operation
that approximates the centered Fourier transform and $\| \tilde F^{\rm d}_{\rm c} -  F^{\rm d}_{\rm c}\| = O(\epsilon)$.
We define the coefficients $\alpha_l$ and $\beta_l$ and
the operations $B_e$ and $\hat B$ such that
\begin{align}
\ket \psi &= B_e \ket 0 =\sum_{l=-L}^{L-1} \sqrt{\alpha_l} \ket l \; , \\
\label{eq:approxQFT}
\tilde F_{\rm c}^{\rm d} \ket \phi &= \hat B \ket 0   =\sum_{l=-L}^{L-1} \sqrt{\beta_l/2} \ket l \;.
\end{align}
The state $\ket 0$ is the initial state of $m=\log_2 (2L)$ qubits.
Then,
\begin{align}
&\bra 0 (B_e^\dagger \otimes \one) {\rm select}(\bar U) (B_e \otimes \one) \ket 0 =\sum_{l=-L}^{L-1} \alpha_l U_l \;, \\
&\bra 0 (\hat B^\dagger \otimes \one) {\rm select}(\bar U) (\hat B \otimes \one) \ket 0 =\sum_{l=-L}^{L-1} |\beta_l| U_l/2 \;.
\end{align}
The operation ${\rm select}(\bar U) $ is unitary and was defined in Eq.~\eqref{eq:selectU}.
Since $\| \ket \psi \|=O(1)$, $\|{\rm select}(\bar U) \|=1$, $\| (B_e-\hat B) \ket 0\|=O(\epsilon)$, and $\| \bra 0 (B_e-\hat B)^\dagger \|=O(\epsilon)$,
we obtain
\begin{align}
\| \bra 0 (B_e^\dagger \otimes \one) {\rm select}(\bar U) (B_e \otimes \one) \ket 0 - \\
\nonumber
- \bra 0 (\hat B^\dagger \otimes \one) {\rm select}(\bar U) (\hat B \otimes \one) \ket 0 \| = \\
\nonumber
=\| \sum_{l=-L}^{L-1} (\alpha_l -|\beta_l|/2) U_l \| = O(\epsilon)\;.
\end{align}
This proves Eq.~\eqref{eq:cor2}.

\section{Centering states}
\label{app:centering}

\begin{figure}
	\begin{center}
		\includegraphics[width=.25\textwidth]{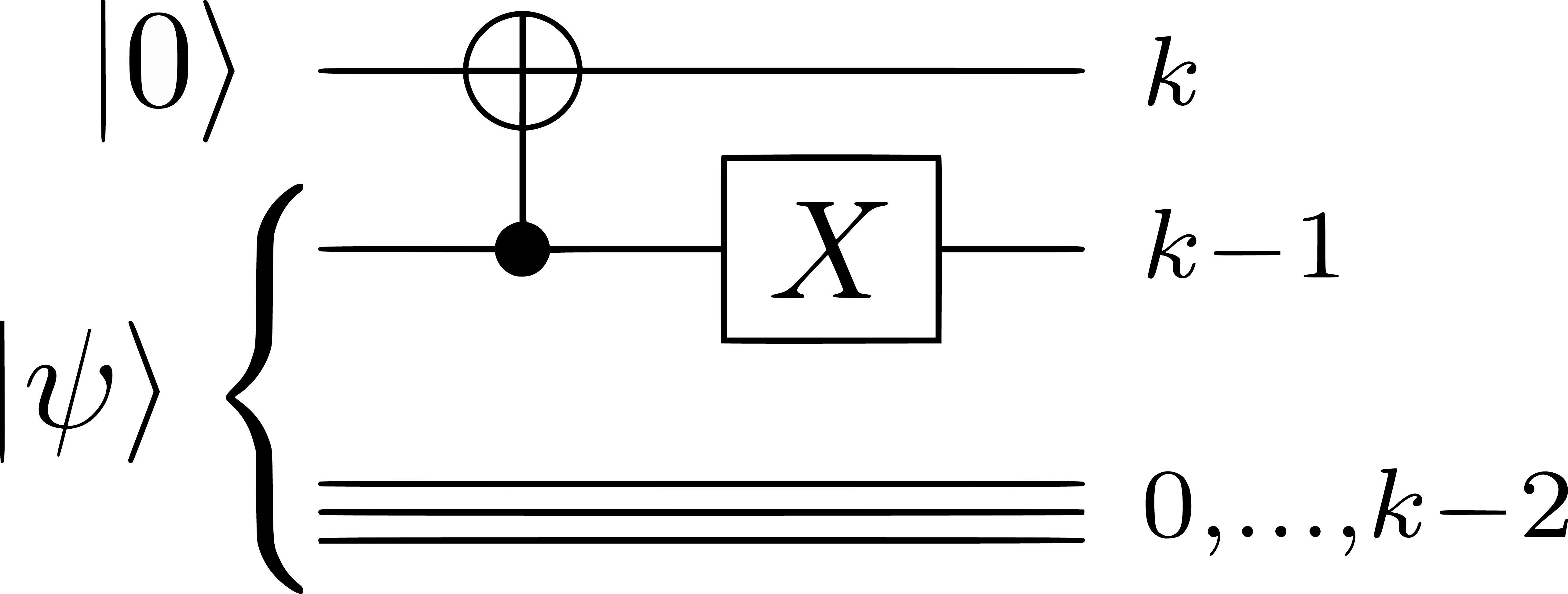}
	\end{center}
	\caption{Quantum circuit for centering a $ k $-qubit state on a register of $ k+1 $ qubits. The wires are arranged from top to bottom in the order of the most to least significant bit in the binary encoding. Only the top two wires are acted upon; the lower $ (k-1) $ wires, labelled 0 through $ k-2 $, are unaffected.}
	\label{fig:centeringcircuit}
\end{figure}

\begin{figure}
	\centering
	\subfloat[]{
		\begin{tikzpicture}
		\draw (0,0)--(0.6,0) node[pos=0.6, below]{$\scriptstyle |00\rangle$} ;\draw (0.8,0)--(1.4,0) node[pos=0.6, below]{$\scriptstyle |01\rangle$}; \draw (1.6,0)--(2.2,0) node[pos=0.6, below]{$\scriptstyle |10\rangle$}; \draw (2.4,0)--(3.0,0) node[pos=0.6, below]{$\scriptstyle |11\rangle$};
		\draw [line width=1pt] (0.3,0.1)--(0.3,0.5); \draw [line width=1pt] (1.1,0.1)--(1.1,0.9); \draw [line width=1pt] (1.9,0.1)--(1.9,0.7);\draw [line width=1pt] (2.7,0.1)--(2.7,0.8);
		\end{tikzpicture}
	}
	\par
	\subfloat[]{
		\begin{tikzpicture}
		\draw (0,0)--(0.6,0) node[pos=0.6, below]{$\scriptstyle |000\rangle$} ;\draw (0.8,0)--(1.4,0) node[pos=0.6, below]{$\scriptstyle |001\rangle$}; \draw (1.6,0)--(2.2,0) node[pos=0.6, below]{$\scriptstyle |010\rangle$}; \draw (2.4,0)--(3.0,0) node[pos=0.6, below]{$\scriptstyle |011\rangle$};
		\draw (3.2,0)--(3.8,0) node[pos=0.6, below]{$\scriptstyle |100\rangle$} ;\draw (4.0,0)--(4.6,0) node[pos=0.6, below]{$\scriptstyle |101\rangle$}; \draw (4.8,0)--(5.4,0) node[pos=0.6, below]{$\scriptstyle |110\rangle$}; \draw (5.6,0)--(6.2,0) node[pos=0.6, below]{$\scriptstyle |111\rangle$};
		\draw [line width=1pt] (0.3,0.1)--(0.3,0.5); \draw [line width=1pt] (1.1,0.1)--(1.1,0.9); \draw [line width=1pt] (1.9,0.1)--(1.9,0.7);\draw [line width=1pt] (2.7,0.1)--(2.7,0.8);
		\end{tikzpicture}
	}
	\par
	\subfloat[]{
		\begin{tikzpicture}
		\draw (0,0)--(0.6,0) node[pos=0.6, below]{$\scriptstyle |000\rangle$} ;\draw (0.8,0)--(1.4,0) node[pos=0.6, below]{$\scriptstyle |001\rangle$}; \draw (1.6,0)--(2.2,0) node[pos=0.6, below]{$\scriptstyle |010\rangle$}; \draw (2.4,0)--(3.0,0) node[pos=0.6, below]{$\scriptstyle |011\rangle$};
		\draw (3.2,0)--(3.8,0) node[pos=0.6, below]{$\scriptstyle |100\rangle$} ;\draw (4.0,0)--(4.6,0) node[pos=0.6, below]{$\scriptstyle |101\rangle$}; \draw (4.8,0)--(5.4,0) node[pos=0.6, below]{$\scriptstyle |110\rangle$}; \draw (5.6,0)--(6.2,0) node[pos=0.6, below]{$\scriptstyle |111\rangle$};
		\draw [line width=1pt] (1.9,0.1)--(1.9,0.5); \draw [line width=1pt] (2.7,0.1)--(2.7,0.9); \draw [line width=1pt] (3.5,0.1)--(3.5,0.7);\draw [line width=1pt] (4.3,0.1)--(4.3,0.8);
		\end{tikzpicture}
	}
	\caption{Example of centering a 2-qubit state on a register consisting of 3 qubits (the vertical bars represent probability amplitudes): (a) shows the initial quantum state on 2 qubits, (b) is the state when a new qubit initialized in $ |0\rangle $ is added to the register, and (c) is after the action of the circuit in Fig.\ \ref{fig:centeringcircuit}}
	\label{fig:centeringexample}
\end{figure}
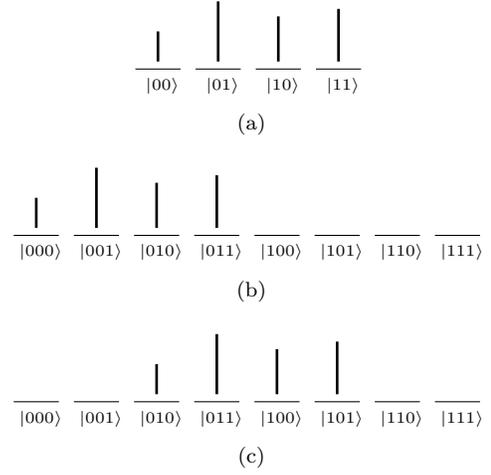

Given a quantum state $\ket{\phi_k}$ over $2L^*$ basis states on a register of $k=\log_2(2L^*)$ qubits, we need to center it over $2L>2L^*$ basis states encoded on a register of $ m=\log_2(2L) $ qubits. Centering is a transformation from
 the old to the new register that takes the basis states $\left\{\ket{j}| 0\leq j\leq 2L^*-1\right\}$ of the old register to the basis states $\left\{\ket{j}| (L-L^*)\leq j\leq (L+L^*)-1\right\}$ of the new register. To see how this can be done, we consider the case when we are given a state prepared on a register of $k$ qubits and wish to center it on one consisting of $(k+1)$ qubits. Suppose that the $ k $-qubit register is in a basis state $\ket{q_{k-1}q_{k-2}\dots q_0}$ with the qubits labelled from 0 to $(k-1)$, in the order of the most to least significant bits in the corresponding binary string. $q_j$ here denotes the value (0 or 1) of the state of qubit $j$. The decimal number represented by this binary bit-string is $q_{k-1}\cdot 2^{k-1}+q_{k-2}\cdot 2^{k-2}+\dots+q_0\cdot2^0$. We append an additional qubit, initialized in 0, to the left, i.e., $\ket{0_{k}q_{k-1}q_{k-2}\dots q_0}$. The centering is a permutation of the bases that corresponds to a cyclic shift by $2^{k-1}$. Figure \ref{fig:centeringcircuit} shows the quantum circuit that implements this permutation, and Fig.\ \ref{fig:centeringexample} demonstrates its action for $ k=2 $.  We first perform a CNOT gate where the target qubit is the appended qubit and the control qubit is the $k$th qubit. This copies $q_{k-1}$ to the $ k $-th position, the decimal number represented by the new bit-string being $q_{k-1}\cdot2^{k}+q_{k-1}\cdot2^{k-1}+\dots+q_0\cdot2^0$. Finally, we perform a NOT (i.e., Pauli $X$) gate on the $(k-1)$-th qubit, preparing a quantum state which now represents the decimal number $q_{k-1}\cdot2^{k}+(q_{k-1}\oplus 1)\cdot2^{k-1}+\dots+q_0\cdot2^0$. The difference between the initial and final decimal numbers associated
 with the states of $k+1$ and $k$ qubits is
\begin{align}
q_{k-1}\cdot2^{k}+(q_{k-1}\oplus 1)\cdot2^{k-1} - q_{k-1}\cdot 2^{k-1} \nonumber \\ = 2^{k-1}\;,
\end{align}
which is the desired shift.
Note that the state of the other $(k-1)$ qubits remains unaffected. To obtain a centered state in a register of $m$ qubits, it suffices to repeat the above procedure for each appended qubit, i.e., $(m-k)$ times, which requires $(m-k)$ two qubit gates overall.

\end{document}